\pdfoutput=1
\RequirePackage{ifpdf}
\ifpdf 
\documentclass[pdftex]{sigma}
\else
\documentclass{sigma}
\fi

\def\phi{\varphi}

\def\ang{\gamma}

\def\Scal{{\mathcal S}}
\def\Ucal{{\mathcal U}}
\def\Hcal{{\mathcal H}}

\def\Dcal{{\mathcal D}}
\def\tr{\operatorname{tr}}

\def \CC{{\mathcal C}}

\def\le{\left}
\def\ri{\right}

\def \d{{\mathrm d}}
\def\res{\mathop{\mathrm {res}}}

\def\Psih{\widehat{\Psi}}
\def\Phih{\widehat{\Phi}}

\def\be{\begin{equation}}
\def\ee{\end{equation}}

\def\Hcal{{\mathcal H}}

\def\Acal{\mathcal A}

\def \eqref #1{(\ref{#1})}
\def \1{\mathbf 1}

\def\vt\tilde{v}

\def\C{{\mathbb C}}

\def\Z{{\mathbb Z}}

\def\a{\alpha}
\def\g{\gamma}
\def\b{\beta}

\def\l{\lambda}

\def\s{\sigma}
\def\p{\partial}
\def\Ccal{{\mathcal C}}
\def\Ch{{\widehat{{\mathcal C}}}}

\def\pa{\partial}

\def\f{\frac}
\def\Scal{{\mathcal S}}

\numberwithin{equation}{section}

\newtheorem{Theorem}{Theorem}[section]
\newtheorem*{Theorem*}{Theorem}
\newtheorem{Corollary}[Theorem]{Corollary}
\newtheorem{Lemma}[Theorem]{Lemma}
\newtheorem{Proposition}[Theorem]{Proposition}
 { \theoremstyle{definition}
\newtheorem{Definition}[Theorem]{Definition}

\newtheorem{Remark}[Theorem]{Remark} }

\begin{document}

\allowdisplaybreaks

\renewcommand{\thefootnote}{}

\newcommand{\arXivNumber}{2303.05602}

\renewcommand{\PaperNumber}{104}

\FirstPageHeading

\ShortArticleName{Szeg\H{o} Kernel and Symplectic Aspects of Spectral Transform}

\ArticleName{Szeg\H{o} Kernel and Symplectic Aspects of Spectral\\ Transform for Extended Spaces of Rational Matrices\footnote{This paper is a~contribution to the Special Issue on Evolution Equations, Exactly Solvable Models and Random Matrices in honor of Alexander Its' 70th birthday. The~full collection is available at \href{https://www.emis.de/journals/SIGMA/Its.html}{https://www.emis.de/journals/SIGMA/Its.html}}}

\Author{Marco BERTOLA, Dmitry KOROTKIN and Ramtin SASANI}
\AuthorNameForHeading{M.~Bertola, D.~Korotkin and R.~Sasani}

\Address{Department of Mathematics and Statistics, Concordia University,\\
1455 de Maisonneuve W., Montr\'eal, H3G 1M8 Qu\'ebec, Canada}
\Email{\!\href{mailto:marco.bertola@concordia.ca}{marco.bertola@concordia.ca},\! \href{mailto:dmitry.korotkin@concordia.ca}{dmitry.korotkin@concordia.ca},\! \href{mailto:ramtin.sasani@concordia.ca}{dramtin.sasani@concordia.ca}}

\ArticleDates{Received March 14, 2023, in final form December 02, 2023; Published online December 22, 2023}

\Abstract{We revisit the symplectic aspects of the spectral transform for matrix-valued rational functions with simple poles. We construct eigenvectors of such matrices in terms of the Szeg\H{o} kernel on the spectral curve. Using variational formulas for the Szeg\H{o} kernel we construct a new system of action-angle variables for the canonical symplectic form on the space of such functions. Comparison with previously known action-angle variables shows that the vector of Riemann constants is the gradient of some function on the moduli space of spectral curves; this function is found in the case of matrix dimension 2, when the spectral curve is hyperelliptic.}

\Keywords{spectral transform; Szeg\H{o} kernel; variational formulas}

\Classification{53D30; 34M45}

\renewcommand{\thefootnote}{\arabic{footnote}}
\setcounter{footnote}{0}

\section{Introduction}

The spectral transform for matrix-values functions rationally depending on a spectral parameter
is at the core of the modern theory
of integrable systems (see
the textbook \cite{BBT} and references therein
for a detailed presentation of the subject and its history). Such a transform admits a~far-reaching generalization to the space of Higgs fields on curves of higher genus, which is central in the theory of Hitchin's systems \cite{DonMar,Hitchin1,Hitchin2}. The infinite dimensional analog of this transform for integrable PDEs of Korteweg--de Vries type is the starting point of the theory of soliton equations developed over the last 50 years \cite{ZMNP}.

The spectral transform can be shown to be symplectic with respect to natural symplectic structures on the target and the source. This leads to the Liouville integrability of the natural Hamiltonian flows (thus implying the name ``integrable systems''),
see \cite{BBT,FadTak}. The integrability also plays the main role in quantization.

In this paper, we analyze symplectic aspects of the spectral transform for matrix-valued rational functions with simple poles
which is the case of primary importance for finite-dimensional integrable systems.
We introduce the extension of the standard spectral transform, following the recent paper
\cite{CMP} on the symplectic theory of the monodromy map for Fuchsian connections.
The first main object is the phase space $\Acal$ defined as follows:
\begin{gather}\label{Acal11}
\Acal=\Biggl\{\bigl\{G_j,L_j\bigr\}_{j=1}^m ,\, \sum_{j=1}^m G_j L_j G_j^{-1} =0\Biggr\}\Big/{\sim},
\end{gather}
where $G_j\in {\rm SL}(n)$ and $L_j\in \mathfrak{sl}(n)$; $L_j$ are diagonal matrices such that all eigenvalues of each $L_j$ are different. The equivalence relation $\sim$
identifies the sets of $G_j$'s which differ by a simultaneous left action $G_j\to {\rm SG}_j$ with $S\in {\rm SL}(n)$
($L_j$ are kept unchanged).
Let
\begin{gather*}
A(z)=\sum_{j=1}^m \f{G_j L_j G_j^{-1}}{z-t_j},
\end{gather*}
where $t_j\neq t_k$ is a set of poles.
Denote by $\Dcal$ the locus of codimension 1 in $\Acal$ defined by the condition that the spectral curve
\begin{gather}\label{charpo}
 \det \le(A(z)-y\1\ri)=0
\end{gather}
 is non-smooth or reducible.

The space $\Acal$ \eqref{Acal11} is equipped with the following symplectic form \cite{CMP}:
\begin{gather}\label{sf}
\omega_\Acal = \delta \theta_\Acal=
\tr \sum_{j=1}^m \bigl\{ \delta G_j G_j^{-1}\wedge L_j \delta G_j G_j^{-1}+ \delta L_j \wedge G_j^{-1} \delta G_j\bigr\},
\end{gather}
where the symplectic potential is given by
\begin{gather*}
\theta_\Acal=\tr \sum_{j=1}^m L_j G_j^{-1} \delta G_j
\end{gather*}
(clearly, $\delta \theta_\Acal =\omega_\Acal$).
The form $\omega_\Acal$ \eqref{sf} is invariant under the simultaneous transformation $G_j\to {\rm SG}_j$, and, therefore, it is indeed the form
on $\Acal$. Moreover, the form $\omega_\Acal$ is {\it non-degenerate} on $\Acal$ \cite{CMP}.

The second main object is the space of spectral data $\Scal$ defined as follows:
\begin{gather}\label{SP}
\Scal=\bigl\{ \{Q_j\}_{j=2}^n,\, q\in \C^g, \,\{R_j\in \operatorname{Diag} ({\rm SL}(n))\}_{j=1}^m, \,\mathfrak t \bigr\},
\end{gather}
where $Q_j(z)$ is an arbitrary rational function with poles of order $j$ at the points $t_k$ such that
the $j$-differential $Q_j(z) (\d z)^j$ is holomorphic at $z=\infty$; $q\in \C^g$ and $R_1,\dots, R_m$ are ${\rm SL}(n)$ diagonal matrices.
We require the functions $Q_j$ satisfy the condition that the curve $\CC$ defined via
\begin{gather}\label{fyz}
y^n+Q_{2}(z) y^{n-2}+\dots +Q_n(z)=0
\end{gather}
is smooth and non-reducible, $\mathfrak t$ in \eqref{SP} denotes a Torelli marking of the curve $\CC$, i.e., the choice of canonical basis of cycles $\bigl\{a_\a,b_\a\bigr\}_{\a=1}^g$ with intersection pairing $a_i\circ b_j=\delta_{ij}$, where $g$ is the genus~of~$\CC$.

The direct spectral transform $F\colon \Acal\setminus\Dcal\to\Scal$ depending on $m$ points $t_j\in \C$ and the Torelli marking of the spectral curve
 is defined as follows.
Given a point in the space $\Acal$, we define
\begin{gather*}
A_j=G_j L_j G_j^{-1}
\end{gather*}
(according to the definition of the space $\Acal$ we have $\sum_{j=1}^m A_j=0$)
and construct the rational matrix
\begin{gather*}
A(z)=\sum_{j=1}^m \f{A_j}{z-t_j}.
\end{gather*}

The meromorphic function $Q_j$ is then defined to be the coefficient in front of $y^{n-j}$ of the characteristic polynomial of the matrix $A(z)$.
In particular, $Q_1(z)=\tr A(z)=0$. Define now the irreducible and smooth curve $\CC$ by the equation \eqref{charpo}.
The genus of $\CC$ equals
\begin{gather*}
g=\f{n(n-1)}{2}m+1-n^2.
\end{gather*}
The vector $q$ is constructed via the Abel map of a divisor of degree $g-1$ associated to normalized eigenvectors of the matrix $A(z)$ (see Section \ref{DST} for details). Starting from some distinguished set of eigenvectors of $A_j$ the diagonal matrices $R_j$ determine the normalization of an arbitrary set of eigenvectors relative to this distinguished basis, see Section \ref{DST} for details.

The inverse spectral map $F^{-1}$ can be explicitly described as follows.
Consider a point of the space $\Scal$ \eqref{SP}, choose a basepoint
$z_0\in\C$, and define the spectral curve $\CC$ via \eqref{fyz}.
Introduce the canonical meromorphic abelian differential $v=y\,\d z$ on $\CC$.
Denote by $\Theta(\cdot)$ the theta-function on the Jacobian of
$\CC$ corresponding to the Torelli marking $\mathfrak t$
and define the codimension one locus $(\Theta)\subset \Scal$ by the condition $\Theta(q)=0$.
The main ingredient of the inverse spectral transform
\begin{gather*}
F^{-1}\colon \ \Scal\setminus (\Theta)\to \Acal
\end{gather*}
is the Szeg\H{o} kernel defined by the formula \cite{Fay73}
\begin{gather}\label{Szegodef}
S_q(x,y)= \f{\Theta(\Ucal(x)-\Ucal(y)+q)}{\Theta(q) E(x,y)},
\end{gather}
where $E(x,y)$ is the prime-form, $\Ucal(x)$ is the Abel map such that
$\bigl(\Ucal(x)-\Ucal(y)\bigr)_\a=\int_{y}^x v_\a$ ($v_\a$~is the holomorphic abelian differential on $\CC$ normalized by the property $\oint_{a_\beta} v_\a=\delta_{\a\b}$)
 and $\Theta$ is the Riemann theta-function on the curve \eqref{fyz}, $q\in \C^g$.
The Szeg\H{o} kernel is the reproducing kernel in the line bundle of degree $g-1$ over $\CC$ defined by the vector $q$, see \cite[formula~(2.5)]{Fay92} discussion around this formula.
The main analytical tools used in this paper are the variational formulas for the Szeg\H{o} kernel which we derive on the basis of results of \cite{IMRN,KalKor,JDG}.

Denote by $z^{(
{k})}$ the point on $
{k}$-th sheet of $\CC$ having projection $z\in \C$. Now for a point of the space $\Scal\setminus (\Theta)$, we define
\begin{gather*}
\bigl(L_j\bigl)_{kk}={\rm res}\big|_{t_j^{(k)}}v, \qquad
G_j= \Psih(t_j) R_j,
\end{gather*}
where
\begin{gather}\label{psihint}
\Psih_{\ell k}(z)=\psi\bigl(z^{(k)}, z_0^{(\ell)}\bigr), \qquad \psi(x,x_0)=S_{q}(x,x_0) \f{z(x)-z(x_0)}{\sqrt{\d z(x)}\sqrt{\d z(x_0)}}.
\end{gather}
The function \eqref{psihint} was used in \cite{Annalen} to solve an arbitrary Riemann--Hilbert problem with quasi-permutation monodromy matrices; various ingredients of the construction of \cite{Annalen} appeared already in the 1987 paper by Knizhnik \cite[Section IV]{Knizhnik}. The idea to use \eqref{psihint} for diagonalization of rational matrices goes back to \cite{Borot_Eynard}, where it was used to reformulate the standard formalism of matrix Baker--Akhiezer function
{ that was pioneered in \cite{Krichever77} and generalized to the higher genus case in \cite{Krichever} (see also the further history in the textbook \cite{BBT})}.

The natural coordinates on $\Scal$ are defined as follows. These are the $a$-periods $I_\a=\oint_{a_\a} v$ of the differential $v=y\,\d z$, the components of the vector $q$,
the variables
\smash{$\bigl\{\mu_j^{(k)}\bigr\}$} for $j=1,\dots,m$ and $k=1,\dots,n-1$ such that
\begin{gather*}
\mu_j^{(k)}-\mu_j^{(k-1)}={\rm res}|_{t_j^{(k)}} v,\qquad k=1,\dots,n-1,
\end{gather*}
(with the understanding that $\mu_j^{(0)}=\mu_j^{(n)}=0$).
The toric variables $\rho_j^{(k)}$ parametrize the diagonal matrices $R_j$ with $r_j^{(k)}=(R_j)_{kk}$ via
\begin{gather*}
\log r_j^{(k)}=\rho_j^{(k)}-\rho_j^{(k-1)}.
\end{gather*}

Let us now describe the main results of this paper.

First, we prove that the symplectic potential $\theta_\Acal$ and the associate symplectic form $\omega_\Acal$ can be written as follows:
\begin{gather*}
\theta_{\Acal}= \sum_{\a=1}^g I_\a \delta q_\a + \sum_{j=1}^m\sum_{k=1}^{n-1} \mu_j^{(k)}\delta \rho_j^{(k)},
\\
\omega_{\Acal}= \sum_{\a=1}^g \delta I_\a \wedge \delta q_\a
+ \sum_{j=1}^m\sum_{k=1}^{n-1} \delta\mu_j^{(k)} \wedge \delta\rho_j^{(k)}.
\end{gather*}

Let us now introduce the space $\Acal_0$
\begin{gather*}
\Acal_0=\Biggl\{\bigl\{A_j\bigr\}_{j=1}^m,\ \sum_{j=1}^m A_j =0\Biggr\}\Big/{\sim}
\end{gather*}
such that all $A_j$ are diagonalizable, and the equivalence relation $\sim$ identifies the sets of $A_j$ which differ by a simultaneous conjugation $A_j\to SA_j S^{-1}$ with $S\in {\rm SL}(n)$. There is a trivial map $\Acal\to \Acal_0$ which ``forgets'' about the toric variables, i.e., about the normalization of the eigenvectors of $A_j$.

In particular, we can fix the matrices $L_j$'s and perform the symplectic reduction on the corresponding level sets \smash{$\mu_j^{(k)}=\rm const$}. Then we get the second main result of this paper which is the expression for the reduced symplectic form
{on the level sets} within $\Acal_0$, which we denote by $\Acal_0^L$,\looseness=-1
\begin{gather}\label{formleaf}
\omega_{\Acal_0}^L= \sum_{\a=1}^g \delta I_\a \wedge \delta q_\a.
\end{gather}
This result looks similar, but is actually different from the known expression for the form $ \omega^L_{\Acal}$ (see, for example, \cite[formula (5.75)]{BBT}):
\begin{gather}\label{formleaf1}
\omega_{
{\Acal}}^L= \sum_{\a=1}^g \delta I_\a \wedge \delta \ang^x_\a,
\end{gather}
where $\ang^x$ is defined via the Abel map of the ``dynamical divisor'' with basepoint $x$ such that~$z(x)$ is independent of all dynamical variables. The
relationship between the Abel map $\ang^x$ and the vector $q$ is given by
\begin{gather*}
\ang^x=q+ K^x,
\end{gather*}
where $K^x$ is the vector of Riemann constants with the same basepoint $x$.

The coincidence of 2-forms \eqref{formleaf} and \eqref{formleaf1} leads to the third main result of the paper which is the following (somehow surprising) integrability condition for the vector of Riemann constants:
\begin{gather}\label{intK}
\f{\delta K^x_\a}{\delta I_\b}=\f{\delta K^x_\b}{\delta I_\a}
\end{gather}
which implies that $K$ is a gradient of some scalar function of moduli $I_\a$.

The relation \eqref{intK} is unexpected. Here we also verify it directly, using variational formulas for the vector of Riemann constants on the moduli space.

An interesting open problem (which is rather trivial for $n=2$) is to compute the potential for the vector $K^x$ in closed form; such a potential is nothing but the function generating the change of Darboux coordinates from $\bigl(I_\a, q_\a\bigr)$ to $\bigl(I_\a, \ang^x_\a\bigr)$.

To prove the main results of the paper we have derived new variational formulas for the Szeg\H{o} kernel and the vector of Riemann constants with respect to the moduli $I_\alpha$ and the eigenvalues~\smash{$\mu_j^{(k)}$} (Sections \ref{secvarSz} and \ref{VarKsec}).

\section[Spaces A and S and their symplectic geometry]{Spaces $\boldsymbol\Acal$ and $\boldsymbol\Scal$ and their symplectic geometry}

\subsection[The phase space \Acal]{The phase space $\boldsymbol\Acal$}

Consider the phase space $\Acal$ defined as follows:
\begin{gather}\label{Acal}
\Acal=\Biggl\{\bigl\{G_j,L_j\bigr\}_{j=1}^m,\, \sum_{j=1}^m G_j L_j G_j^{-1} =0\Biggr\}\Big/{\sim},
\end{gather}
where $G_j\in {\rm SL}(n)$, $L_j\in \mathfrak{sl}(n)$ are diagonal matrices and the equivalence relation $\sim$
identifies the sets $G_j$'s which differ by a simultaneous multiplication $G_j\to {\rm SG}_j$ with $S\in {\rm SL}(n)$ and
$L_j$ are kept unchanged.

The space $\Acal$ possesses the following symplectic form \cite{CMP}:
\begin{gather}\label{sf1}
\omega_\Acal =
\tr \sum_{j=1}^m \bigl(\delta G_j G_j^{-1}\wedge L_j \delta G_j G_j^{-1}+ \delta L_j \wedge G_j^{-1} \delta G_j\bigr).
\end{gather}
The symplectic form satisfies $
\omega_\Acal= \delta \theta_\Acal$,
where the symplectic potential is given by
\begin{gather}\label{thA}
\theta_\Acal=\tr \sum_{j=1}^m L_j G_j^{-1} \delta G_j.
\end{gather}
The form $\omega_\Acal$ \eqref{sf1} is invariant under the action of simultaneous left multiplication by ${\rm GL}_n$ as follows $G_j\to {\rm SG}_j$, and, therefore, it indeed factors through the quotient by this action to a~form
on $\Acal$. Moreover, the form $\omega_\Acal$ is non-degenerate on $\Acal$: for comparison, the symplectic form discussed in \cite[Section 5.9]{BBT} is given by \eqref{sf1} without the last term and it is non-degenerate only on symplectic leaves $L_j=\rm const$.

Following \cite{CMP}, now we discuss the Poisson structure corresponding to the form $\omega_{\Acal}$.
First, for any matrix $M$, we use the following notation for the Kronecker products:
\begin{gather}\label{tensornotation}
\overset{1}{M} = M\otimes \1, \qquad \overset{2}{M} =\1\otimes M.
\end{gather}

Consider the space of pairs $(G,L)$, where $G\in {\rm SL}(n)$ and $L=\operatorname{diag}\bigl(\l^{(1)},\dots,\l^{(n)}\bigr)$ with
$\sum_{j=1}^n \l^{(j)}=0$ and $\l^{(j)}\neq \l^{(k)}$. Define the symplectic form on this space
\begin{gather}\label{omega0}
\omega_0=
 \tr \bigl(\delta G G^{-1}\wedge L\,\delta G G^{-1}+ \delta L \wedge G^{-1} \delta G\bigr).
\end{gather}
If the eigenvalues $L$ are kept constant (i.e., $ \delta L\equiv 0$) then the form coincides with the form in~\cite{AlekMalk}.
Then the Poisson structure on $\Hcal$ inverse to the symplectic form \eqref{omega0}
 is (see \cite[Propositions~2.2 and~2.3]{CMP})
\begin{gather}\label{b2}
\big\{\overset{1}{G},\overset{2}{G}\big\}_0=-\overset{1}{G}\overset{2}{G} r(L), \qquad
\big\{\overset{1}{G},\overset{2}{L}\big\}_0=-\overset{1}{G} \Omega,
\end{gather}
where
\begin{gather*}
r(L)=\sum_{i<j} \f{E_{ij}\otimes E_{ji}-E_{ji}\otimes E_{ij}}{\l^{(i)}-\l^{(j)}}
\end{gather*}
and
\begin{gather*}
\Omega= \sum_{i=1}^n E_{ii}\otimes E_{ii} - \frac 1 n \1 \otimes \1,
\end{gather*}
we use the standard notation $E_{ij}$ for the matrix with only one non-vanishing element equal to~$1$
in the $(i,j)$ entry. The matrix $r(L)$ is a simplest example of dynamical $r$-matrix \cite{Etingof}.
In the formula \eqref{b2} we have used the notation \eqref{tensornotation} and the expression should be understood as collecting in a compact tensor notation the Poisson brackets for $\bigl\{G_{ab}, G_{cd}\bigr\}_0$ and $\bigl\{G_{ab}, L_{cd}\bigr\}_0$, for all choices of the indices $a$, $b$, $c$, $d$. Explicitly it reads as follows:
\begin{gather*}
\big\{G_{bj} , G_{c\ell } \big\}_0 = \frac{G_{b\ell} G_{cj}}{\lambda_j - \lambda_\ell}, \qquad
\big\{G_{bk}, \lambda_{\ell }\big\}_0 =-G_{bk} \delta_{\ell k}.
\end{gather*}

Theorem 2.1 of \cite{CMP} shows that the bracket \eqref{b2} induces the
Kirillov--Kostant Poisson bracket for $A=GLG^{-1}$. In ${\rm SL}(2)$ case, the bracket \eqref{b2} was first introduced in \cite{AF}.
Without the second term in the symplectic form \eqref{omega0} the form becomes degenerate in the extended space but it is non-degenerate on the leaves given by keeping $L$ constant. In this case, the corresponding Poisson structure on these leaves coincides with the Kostant--Kirillov--Souriau structure, i.e., the Lie--Poisson structure \cite{AlekMalk}.

The Poisson bracket $\{\cdot,\cdot\}_{\Acal}$ on the space $\Acal$ inverse to the form \eqref{sf1} is then given by
\begin{gather}\label{KKex}
\big\{\overset{1}{G_j},\overset{2}{G_j}\big\}_{\Acal} = \big\{\overset{1}{G_j},\overset{2}{G_j}\big\}_{0}, \qquad \big\{\overset{1}{G_j},\overset{2}{L_j}\big\} _{\Acal}= \big\{\overset{1}{G_j},\overset{2}{L_j}\big\} _0
\end{gather}
and all other brackets vanish.
To resolve the condition $\sum_{j=1}^n \l^{(j)}=0$ we shall use the parametri\-za\-tion
\begin{gather}\label{lmu}
\l^{(1)}= \mu^{(1)},\quad
\l^{(2)}= \mu^{(2)}-\mu^{(1)},\quad \dots, \quad \l^{(n-1)}= \mu^{(n-1)}-\mu^{(n-2)}, \quad \l^{(n)}=-\mu^{(n-1)},\!\!\!
\end{gather}
where $\mu^{(j)}$, $j=1,\dots,n-1$ are independent variables such that $\l^{(j)}\neq \l^{(k)}$.

\subsection{Extended space of spectral data}
Let $Q_2(z), \dots, Q_n(z)$ be generic rational functions with poles of order $j$ at the points $t_j$ and such that
the $Q_j(z)\,\d z^j$ are all holomorphic at $\infty\in \mathbb P^1$ with $z$ the affine coordinate.
Consider the affine curve $\CC$
defined by the equation
\begin{gather}\label{fyz1}
f(y,z)=0,\qquad
f(y,z):= y^n+Q_{2}(z) y^{n-2}+\dots +Q_n(z).
\end{gather}
The genericity of the functions $Q_j$ is, in our context, defined as irreducibility and smoothness of the curve $\CC$. On $\CC$ we choose a symplectic basis in the homology (i.e., a {\it Torelli marking}) $\mathfrak t =\{a_1,\dots, a_g, b_1,\dots, b_g\}$.
Now we define the extended space of spectral data $\Scal$ as follows:
\begin{gather}\label{SP1}
\Scal=\Big\{ \bigl\{Q_j\bigr\}_{j=2}^n,\, q\in \C^g, \bigl\{R_j\in \operatorname{Diag} \bigl({\rm SL}(n)\bigr)\bigr\}_{j=1}^m, \mathfrak t \Big\},
\end{gather}
where $Q_j(z)$ is an arbitrary rational function with poles of order $j$ at the points $t_j$ such that
the $j$-differential $Q_j(z) (\d z)^j$ is holomorphic at $z=\infty$.
The genus of $\CC$ is given by
\begin{gather}\label{genus1}
g=\f{n(n-1)}{2}m-n^2+1.
\end{gather}

Introduce also the ``theta-divisor'' $\operatorname{div}(\Theta)\subset \Scal$ defined by the equation $\Theta(q)=0$, where $\Theta$ is the Riemann theta-function corresponding to the Torelli marked curve $\CC$.

The term``extended'' in application to the space $\Scal$ refers to addition of diagonal matrices $R_j$ to the
set of differentials $Q_j$ and vector $q$.

The number of branch points of the curve \eqref{fyz1} equals (see \cite[equation~(2.3)]{IMRN}){\samepage
\begin{gather}\label{bp}
p=n(n-1)(m-2)
\end{gather}
which leads to the formula \eqref{genus1} for the genus by Riemann--Hurwitz formula.}

The moduli space of curves of the form \eqref{fyz1} for fixed $t_j$ has dimension $g+m(n-1)$.
The local coordinate system on this space is constructed as follows. First, these are any $m(n-1)$ independent residues of the differential
\begin{gather*}
v=y\,\d z
\end{gather*}
at the points $t_j^{(k)}$,
\begin{gather}\label{deflj}
\lambda_j^{(k)}={\rm res}|_{t_j^{(k)}} v, \qquad k=1,\dots,n-1
\end{gather}
(notice that due to the absence of the term $y^{n-1}$ in \eqref{fyz1} we have
$\sum_{k=1}^n {\rm res}|_{t_j^{(k)}} v=0$ for all $j=1,\dots,m$).

We shall assume that, according to \eqref{lmu},
\begin{gather}\label{lmu1}
\lambda_j^{(k)}=\mu_j^{(k)}-\mu_j^{(k-1)},
\end{gather}
and use the independent variables $\mu_j^{(k)}$, $j=1,\dots,m$, $k=1,\dots,n-1$.

Consider now a basis $\{a_\a,b_\a\}_{\a=1}^g$, $\bigl\{l_j^{(k)}, j=1,\dots, m,\, k=1,\dots,n\bigr\}$ in $H_1\bigl(\CC\setminus \bigl\{ t_j^{(k)}\bigr\},\Z\bigr)$, where~$l_j^{(k)}$ is a small loop around \smash{$t_j^{(k)}$} and $\{a_\a,b_\a\}$ form a symplectic basis in $H_1(\CC,\Z)$.
Consider the $a$-periods of the differential $v$,
\begin{gather}\label{defI}
I_\a=\int_{a_\a} v,\qquad j=1,\dots,g,
\end{gather}
$2\pi {\rm i} \lambda_j^{(k)}$ are periods of $v$ over $l_j^{(k)}$ \eqref{deflj}.
The set of variables $I_j$ \eqref{defI} and $\mu_j^{(k)}$ \eqref{deflj} gives local coordinates on the space of curves \eqref{fyz1}. The full set of local coordinates on the space $\Scal$ \eqref{SP1} is then obtained by adding
a vector $q\in
{\C^g}$ and the variables \smash{$\rho_j^{(k)}$, $j=1,\dots,m$, $k=1,\dots,n-1$} such that
\begin{gather*}
\log r_j^{(k)}=\rho_j^{(k)}-\rho_j^{(k-1)},
\end{gather*}
where $R_j=\operatorname{diag}\bigl(r_j^{(1)},\dots,r_j^{(n)}\bigr)$.

The full set of local coordinates on the space $\Scal$ can be chosen as follows:
\begin{gather}\label{fullco}
\bigl\{ \{q_\a, I_\a \}_{\a=1}^g, \, \bigl\{\mu_j^{(k)}, \rho_j^{(k)}\bigr\}, \, j=1,\dots,m, \, k=1,\dots, n-1\bigl\}.
\end{gather}

The symplectic form on $\Scal$ is then defined as follows:
\begin{gather}\label{oS}
\omega_\Scal=\sum_{\a=1}^g \delta I_\a\wedge \delta q_\a + \sum_{j=1}^m \sum_{k=1}^{n-1}\delta \rho_j^{(k)} \wedge \delta \mu_j^{(k)},
\end{gather}
i.e., $(I_\a, q_\a)$ and $\big(\rho_j^{(k)} , \mu_j^{(k)}\big)$ form canonical pairs of Darboux coordinates.

The Poisson brackets given by the inverse of the form \eqref{oS} look as follows:
\begin{gather*}
\bigl\{I_\a, q_\a\bigr\}_{\Scal}=1,\qquad \bigl\{\rho_j^{(k)} , \mu_j^{(k)}\bigr\}_{\Scal}=1
\end{gather*}
for $\a=1,\dots,g$, $j=1,\dots,m$, $k=1,\dots,n-1$.
All other Poisson brackets are vanishing.

\section{Direct and inverse spectral transform}

\subsection{Direct spectral transform}
\label{DST}

Choose a base point $z_0\in\C$
and a representative $\bigl\{G_j,L_j\bigr\}$ in the equivalence class representing a point of $\Acal$ such that the
matrix
\begin{gather*}
A(z)=\sum_{j=1}^m \f{G_j L_j G_j^{-1}}{z-t_j}
\end{gather*}
satisfies the condition that $A(z_0)$ is diagonal.

Define now the spectral curve $\CC$ by the equation
\begin{gather}\label{charpo1}
 \det \le(A(z)-y\1\ri)=0.
\end{gather}
Denote by $\Dcal$ the locus of codimension 1 in $\Acal$ defined by the condition that the spectral curve~\eqref{charpo1} is non-smooth or reducible (the locus $\Dcal$ depends on positions $t_j$ of the poles).

Let us now define the family of the spectral maps
\begin{gather*}
F^{{\mathfrak t}}\colon \ \Acal\setminus \Dcal \to \Scal
\end{gather*}
 which depends on positions of
poles $t_j$. This map is labelled by the Torelli marking $\mathfrak t$ of the spectral curve.

Now the spectral data (i.e., a point of the space $\Scal$) can be constructed as follows.

\subsubsection[Meromorphic functions Q\_j]{Meromorphic functions $\boldsymbol{Q_j}$}

The standard counting shows that the genus of the algebraic curve \eqref{charpo1}, assuming that it is non-singular and irreducible, is given by \eqref{genus1} and the number $p$ of the branch points (counting with multiplicities) is given by \eqref{bp}.
The symmetric polynomials of $A(z)$ then give the rational functions $Q_j$ in \eqref{SP}.
Since $A(z) = O\bigl(z^{-2}\bigr)$ as $z\to\infty$ (this is due to the condition $\sum_{j=1}^m A_j=0$), we see that
$Q_j$ has zero of degree $2j$ at $z\to\infty$, and, therefore, the differential $Q_j(\d z)^j$ is holomorphic at $z=\infty$, as required in the definition of $\Scal$.

The Torelli marking ${\mathfrak t}$ of $\Scal$ is determined by the labeling of the map $F^{{\mathfrak t}}$.

\subsubsection[Vector q in C\^{}g]{Vector $\boldsymbol{q\in \C^g}$}
Introduce now the adjugate matrix $M(z,y)$ of the degenerate matrix $A(z)-yI$ (remind that~$M$ is the transposed matrix of cofactors of $A(z)-yI$).

Denote the first column of $M(z,y)$ by $\psi(x)$, which is the meromorphic vector-valued function on $\CC$.
Moreover, it is an eigenvector of $A$,
\begin{gather*}
A(z)\psi(x)=y \psi(x),
\end{gather*}
where $x=(z,y)\in \CC$.
The poles of $\psi(x)$ have multiplicity $n-1$. They are at the points $t_j^{(k)}$ for
$j=1,\dots,m$ and $k=1,\dots,n$. Therefore, the total number of poles of each component of $\psi$ on~$\CC$ is equal to $n(n-1)m$.

Moreover, each component of $\psi$ has (due to the condition $\sum A_j=0$) zeros of multiplicity
$2(n-1)$ at each $\infty^j$.

Therefore, the divisor of the $k$th component $\psi_k$ of $\CC$ can be written as
\begin{gather*}
 (\psi_k)=-(n-1)\sum_{j=1}^m\sum_{k=1}^n t_j^{(k)} + 2(n-1)\sum_{k=1}^n \infty^{(k)} + D_k,
\end{gather*}
where $D_k$ is a positive divisor. The degree of $D_k$ thus equals to
\begin{gather*}
\deg D_k=n(n-1)m-2n(n-1)=n(n-1)(m-2)=p=2(g+n-1).
\end{gather*}

The structure of the positive divisor $D_k$ can be further clarified by considering the different normalization of the eigenvector function. Namely, define
\begin{gather}\label{phipsi}
\phi(x)=\frac{\psi(x)}{\psi_1(x)}.
\end{gather}
Then the divisor of the $k$th component can be written as follows:
\begin{gather*}
(\phi_k)=D_k-D_1.
\end{gather*}
It would seem at first that $\phi_k$ is a meromorphic function on $\CC$ of degree
$2(g+n-1)$. However, as we are going to show below, in fact, we have
\begin{gather*}
{\rm deg}\, \phi_k= g+n-1,
\end{gather*}
i.e., the divisor $D_k$ can be represented as a sum of two positive divisors,
\begin{gather}\label{dkt}
D_k= \tilde{D}_0+ \tilde{D}_k,
\end{gather}
where $ \tilde{D}_0$ is independent of $k$, and ${\rm deg} \, \tilde{D}_0={\rm deg }\, \tilde{D}_k=g+n-1$.

Consider the determinant
\begin{gather*}
F(z)=\det \bigl(\phi(z,y_1),\dots,\phi(z,y_n)\bigr)^2.
\end{gather*}
The function $F(z)$ can be considered as a rational function of $z$ on the base $\C$. Its poles have {\it twice} the multiplicity of the pole divisor of the vector $\phi$ on $\CC$. Thus, we can count the degree of the poles of $\phi$ by dividing by two the degree of poles of $F$ on $\C$.
To compute the degree of the polar divisor of $F$ on $\C$, it suffices to compute the degree of its zero divisor.

All zeros of $F(z)$ are at the projections on the base $\C$ of the ramification points on $\CC$ and nowhere else because the vectors $\varphi(z,y_1),\dots, \varphi(z,y_n)$ are linearly independent whenever the eigenvalues are distinct.
Therefore, the number of these zeros equals $p$, and hence the number of poles of $\phi$ on $\CC$ equals $p/2=n(n-1)(m-2)$.

Let us show now that poles of $F$ on $\C$ coincide with the $z$-projection of the poles of some component
of $\phi(x)$ (one needs to show that if $\phi$ has a simple pole at some $x_p=(z_p,y_p) $, then also~$F(z)$ has a double pole $z=z_p$).
According to~\eqref{phipsi}, poles of $\phi$ appear only from zeros of~$\psi_1(x)$. Denote such zero
(assume for simplicity that this zero is simple) by $p=(z_p,y_p)$ and assume that all other components of $\psi$ do not vanish simultaneously at $p$ with the same degree.

Consider the vector $\phi_p= (z-z_p)\phi(x)|_{x=p}$. Values of $\phi(x)$ at all other $n-1$ points
having the same projection $z_p$ on $\C$ are eigenvectors $\phi_p^l$ of the matrix $A(z_p)$ with different eigenvalues. Therefore, the vectors $\phi_p$, $\phi_p^1$, $\dots$, \smash{$\phi_p^{(l-1)}$} are linearly independent and $ F(z) \sim (z-z_p)^{-2} \bigl(C+ O(z-z_p)\bigr)$, thus having pole of the same first order at $z_p$ as $\phi(x)$ itself.

Concluding, the degree of each component $\phi_k$ can not exceed $g+n-1$ while the union of divisors of poles of all components of $\phi$ has degree equal to $g+n-1$, leading to the splitting~\eqref{dkt}.

Since $\phi(x)$ is the eigenvector of $A(z)$ corresponding to the point $x=(z,y)\in \CC$,
\begin{gather*}
A(z) \phi(z,y) = y\phi(z,y),
\end{gather*}
its values $\psi(z,y_j)$, $j=1,\dots,n$ on different sheets of $\CC$ give the full set of eigenvalues of $A(z)$,
\begin{gather*}
A(z) \phi(z,y_j) = y_j(z)\phi(z,y_j).
\end{gather*}

We recall that $\phi_1(x)=1$ by our normalization \eqref{phipsi} for the first component. Since $A(z_0)$ is diagonal and denoting by \smash{$z_0^{(j)}, \ j=1,\dots, n$} the points above $z_0$ on the spectral curve (corresponding to the distinct eigenvalues of $A(z_0)$), it follows that
the eigenvectors \smash{$\psi\bigl(z_0^{(j)}\bigr)$} are proportional to the elementary column vectors \smash{${\rm e}_j=(0,\dots,1,\dots,0)^{\rm T}$} (with $1$ in the $j$th place).
It then follows (due to the first component of $\psi$ in the denominator of $\phi$), that for $z\to z_0$ and $j=2,\dots, n$ we have the asymptotic behaviour
\begin{gather*}\label{phie}
\phi \bigl(z^{(j)} \bigr)= {\rm const} \frac{{\rm e}_j}{z-z_0}+\cdots .
\end{gather*}
Therefore, out of the $g+n-1$ poles of $\phi(x)$ there are
$n-1$ ``trivial'' poles at $z_0^{(2)},\dots,z_0^{(n)}$. Denote the complementary degree $g$ positive divisor by $\hat{D}$,
\begin{gather*}
\hat{D}= {\rm (poles\ of\ } \phi)- \sum_{j=2}^n z_0^{(j)}.
\end{gather*}

Finally, consider the canonical polygon $\hat{C}$ corresponding to the basis of cycles $\{a_j,b_j\}$ and define the vector $q$ as follows:
\begin{gather*}
q=\Ucal_x(\hat{D})-\Acal_x\bigl(z_0^{(1)}\bigr)-K^x,
\end{gather*}
where $\Ucal_x$ is the Abel map, and $K^x$ is the vector of Riemann constants with initial point $x$ corresponding to the same canonical basis.

\subsubsection[Diagonal matrices R\_j (the ``toric variables'')]{Diagonal matrices $\boldsymbol{R_j}$ (the ``toric variables'')}

The choice of matrices $R_j$ is not unique: one can choose an arbitrary reference set $G_j^0$ of eigenvectors of each \smash{$A_j=G_j L_j G_j^{-1}$}
and then define $R_j$ via $G_j=G_j^0 R_j$. The choice of $G_j^0$ which leads to matrices $R_j$ which are natural from the point of view of Poisson geometry is based on relationship of eigenvectors of matrix
$A(z)$ \cite{Borot_Eynard} with solution of the matrix Riemann--Hilbert problem with quasi-permutation monodromies~\cite{Annalen}.

First, introduce the set of generators $\g_1,\dots,\g_p$ of $\pi_1\bigl(CP^1\setminus \bigl\{z_j\bigr\}_{j=1}^p,\, z_0\bigr)$ satisfying the
relation $\g_1\cdots\g_p={\rm id}$. Consider the homomorphism $h$ of this fundamental group to the symmetric group~$S_n$ defined by the covering $\Scal$. Then, since all branch points $z_j$ are assumed to be simple, $h\bigl(\g_j\bigr)$ is a simple permutation of two sheets, which we denote by $l_j$ and $k_j$, i.e., $h\bigl(\g_j\bigr)=\bigl(l_j,k_j\bigr)$.
 Denote the corresponding $n\times n$ permutation matrix by $\sigma_{l_j,k_j}$.

Define monodromy matrices representation by \cite[expression (4.47)]{RH1} (in notations of \cite{RH1} one needs to assume that all $p_\alpha=0$ and all \smash{$r_n^{(k)}=0$}. The monodromy matrices are then matrices of quasi-permutation with positions of non-vanishing entries coinciding with the ones of $\sigma_{l_j,k_j}$). Notice that the branch points $z_j$ are denoted by $\lambda_j$ in~\cite{RH1}.

Denote the solution of the Riemann--Hilbert problem with such monodromies normalized by $\Psi(z_0)=I$ by $\Psi$; this solution is given explicitly by \cite[Theorem 4.7]{RH1}. Define now
\begin{gather}\label{GPsi}
G_j^0=\Psi(t_j).
\end{gather}

According to results of \cite{Borot_Eynard}, the matrix $\Psi(z)$ diagonalizes the matrix $A(z)$. Therefore, matrices~$G_j^0$ diagonalize the matrices $A_j$.

Then matrices $G_j$ from the definition of the space $\Acal$ \eqref{Acal}, which form an arbitrary set of matrices diagonilizing $A_j$,
are related to $G_j^0$ by multiplication with some diagonal matrices $R_j\in {\rm SL}(n)$,
\begin{gather*}
G_j= G_j^0 R_j
\end{gather*}
The matrices $R_j$ are the toric variables from the space $\Scal$ \eqref{SP1}.

\subsection{Inverse spectral transform}

The inverse spectral transform $H$ from $\Scal\setminus (\Theta)$ to $\Acal$ is defined by choosing a basepoint $z_0\in C$. We check {\it a posteriori} that $H$ is independent of the choice of $z_0$.

Starting from a point of the space $\Scal$, we define the curve $\CC$ by equation \eqref{fyz1}, denote by $\pi$ the natural projection $\CC\to C$.
Introduce a fundamental polygon $\Ch$ of $\CC$ and
consider the Szeg\H{o}
kernel $S_q(x,y)$ \eqref{Szegodef}.
We shall denote by $z^{(k)}$ the point on $k$th sheet of $\CC$.

Consider a point $z_0\in C$ and its neighbourhood $D\subset C$ such that no branch point of $\CC$ projects in $D$.
Then
\begin{gather*}
\pi^{-1}(D)= D^{(1)}\cup\dots \cup D^{(1)},
\end{gather*}
where $D^{(j)}\subset\CC$ is a disk on $j$th sheet of $\CC$.
Assuming that $q\not\in (\Theta)$, i.e., $\Theta(q)\neq 0$, consider the following function (in the fundamental polygon $\Ch$ of $\Ccal$), which is holomorphic with respect to $x$ for $x\in \Ch$:
\begin{gather}\label{psidef}
\psi(x,x_0)=S_{q}(x,x_0) E_0\bigl(\pi(x),\pi(x_0)\bigr),
\end{gather}
where $x_0, x\in \CC $, $z=\pi(x)$, $z_0=\pi(x_0)$ and $E_0(z,z_0)$ is the prime-form on $C$,
\begin{gather*}
E_0(z,z_0)= \f{z-z_0}{\sqrt{\d z}\sqrt{\d z_0}}
\end{gather*}
On $\CC$ $\psi(x,x_0)$ has non-trivial holonomies with respect to $x$ along $b$-cycles (we consider $x$ to be the argument of $\psi(x,x_0)$ and $x_0$ as the basepoint), given by ${\rm e}^{2\pi {\rm i} q_j}$.

Let us now fix some sheet, say sheet 1, and define the following meromorphic function of degree $g+1$ on $\CC$:
\begin{gather*}
\varphi(x,x_0)= \f{\psi(x,x_0)}{\psi(x,z_0^{(1)})}=\f{S_{q}(x,x_0)}{S_{q}(x,z^{(1)})}.
\end{gather*}
Remind that $z_0=\pi(x_0)$. For $x_0\neq z_0^{(1)}$ $\varphi(x,x_0)$ has a pole at the pole of the numerator (i.e., at $x=x_0$) and $g$ poles at zeros of the denominator, i.e., points, where $\theta\bigl(A(x)-A\bigl(z_0^{(1)}\bigr)-q\bigr)=0$.

Consider also the column vector $\Psi$ given by
\begin{gather*}
\Psi(x,z_0)=\bigl(\psi\bigl(x,z_0^{(1)}\bigr),\dots, \psi\bigl(x,z_0^{(n)}\bigr)\bigr)^{\rm T}
\end{gather*}
and the corresponding column vector meromorphic function
\begin{gather*}
\Phi(x,z_0)=\f{\Psi(x,z_0)}{\psi\bigl(x,z_0^{(1)}\bigr)} = \bigl(\varphi\bigl(x,z_0^{(1)}\bigr),\dots, \varphi\bigl(x,z_0^{(n)}\bigr)\bigr)^{\rm T}
 = \left(1,\f{S_{q}\bigl(x,z_0^{(2)}\bigr)}{S_{q}\bigl(x,z_0^{(1)}\bigr)} ,\dots,
\f{S_{q}\bigl(x,z_0^{(n)}\bigr)}{S_{q}\bigl(x,z_0^{(1)}\bigr)}\right)^{\rm T}\!.
\end{gather*}

From the vectors $\Psi$ and $\Phi$, we construct also the matrices $\Psih$ and $\Phih$ via
\begin{gather}\label{psih}
\Psih(z,z_0)=\bigl(\Psi\bigl(z^{(1)},z_0\bigr),\dots,\Psi\bigl(z^{(n)},z_0\bigr)\bigr)
\end{gather}
and
\begin{gather}\label{phih}
\Phih(z,z_0)=\bigl(\Phi\bigl(z^{(1)},z_0\bigr),\dots,\Phi\bigl(z^{(n)},z_0\bigr)\bigr).
\end{gather}

These two matrices are related by the diagonal matrix
\begin{gather*}
\Phih(z,z_0)=\Psih(z,z_0) \operatorname{diag}\bigl( \psi^{-1}\bigl( z^{(1)},z_0^{(1)}\bigr), \dots, \psi^{-1}\bigl( z^{(n)},z_0^{(1)}\bigr) \bigr).
\end{gather*}
Define also the diagonal matrix
\begin{gather}\label{Lz}
Y(z)=\operatorname{diag}\bigl(y^{(1)}(z),\dots,y^{(n)}(z)\bigr).
\end{gather}

Now we are in position to formulate the following proposition (this statement is contained in~\cite[Corollary 3.3]{Borot_Eynard}; here we give an independent proof).
\begin{Proposition}
The matrix $A(z)$ defined by
\begin{gather}\label{AzP}
A(z)=\Psih(z) Y(z) \Psih(z)^{-1}
\end{gather}
or, equivalently, by
\begin{gather*}
A(z)=\Phih(z) Y(z) \Phih(z)^{-1},
\end{gather*}
where $\Psih$ is given by $\eqref{psih}$, $\Phih$ is given by $\eqref{phih}$ and $Y$ is given by $\eqref{Lz}$,
is a meromorphic function of $z$ with simple poles at $z_j$, i.e., it can be represented in the form
\begin{gather*}
A(z)=\sum_{j=1}^m\f{A_j}{z-t_j}
\end{gather*}
with
\begin{gather*}
A_j=G_j^0 L_j \bigl(G_j^0\bigr)^{-1},
\end{gather*}
where
\begin{gather*}
L_j={\rm res}|_{t_j} Y ,\qquad j=1,\dots,m
\end{gather*}
and $G_j^0$ given by
\begin{gather}\label{G0def}
(G_j^0)_{ab}=\psi\bigl(t_j^{(a)}, z_0^{(b)}\bigr),
\end{gather}
where $\psi$ is defined in terms of Szeg\H{o} kernel by $\eqref{psidef}$.
\end{Proposition}

\begin{proof} To show that $A(z)$ is a function on the base (i.e., a single-valued function of $z$) consider the analytic continuation of $A(z)$ along a path going around some branch point $z_j=\pi(x_j)$. Then some columns (say, number $l$ and $k$) of $\Psih$ interchange, i.e., it
transforms as $\Psih\to \Psih \Pi_{kl}$, where $\Pi_{kl}$ is the permutation matrix. Simultaneously, the diagonal entries $l$ and $k$ of $Y$ also get
interchanged, i.e., it transforms as $Y\to \Pi_{kl} Y\Pi_{kl}$ (notice that $\Pi_{kl}^2=I$). Thus, the matrix $A(z)$ defined by
\eqref{AzP}, remains invariant, and is a function of $z$. Simple poles of $A(z)$ at $t_j$ arise due to simple poles of function $y$ at the points $t_j^{(a)}$.
\end{proof}

Now we define the inverse spectral transform $H$ as follows.

\begin{Definition}
For a given point of the space $\Scal$ \eqref{SP1}, consider the curve $\CC$ \eqref{fyz1} and let
 $L_j$ be given by \eqref{Lz}. Define also
\begin{gather}\label{Gj}
G_j=G_j^0 R_j,\qquad j=1,\dots,m,
\end{gather}
where $G_j^0$ are defined by \eqref{G0def}.
\end{Definition}

The self-consistency of this definition is proven in the next proposition.

\begin{Proposition}
The set $\bigl\{G_j,L_j\bigr\}_{j=1}^m$ defines an element of the space $\Acal$, in particular,
\begin{gather}\label{GL}
\sum_{j=1}^m G_j L_j G_j^{-1}=0.
\end{gather}
\end{Proposition}
\begin{proof}
To prove \eqref{GL}, we notice that $\infty^{(j)}$ are regular points of the differential $v=y\,\d z$. Therefore,
${\rm res}|_{\infty} A(z)\,\d z=0$, which implies $\sum_{j=1}^m A_j=0$.
\end{proof}

In the sequel, we shall need the following lemma (see \cite[formula~(4.14)]{Annalen}).
\begin{Lemma}
We have
\begin{gather}\label{Psihi}
\Psih^{-1}(z,z_0)=\Psih(z_0,z)
\end{gather}
and
\begin{gather*}
\bigl((G^0_j)^{-1}\bigr)_{ab}=\psi\bigl(z_0^{(b)},t_j^{(a)}\bigr).
\end{gather*}
\end{Lemma}

Finally, we formulate the following proposition.

\begin{Proposition}
The inverse spectral transform $H\colon\Scal\to \Acal$ is the inverse of the direct spectral transform $($depending on Torelli marking of the spectral curve ${\bf t})$ $F^{\bf t}\colon\Acal\to\Scal$,
\begin{gather*}
F^{{\bf t}}\circ H=H\circ F^{{\bf t}}={\rm id}.
\end{gather*}
\end{Proposition}
\begin{proof} The non-trivial part is to check that the toric variables $R_j$ forming the part of the space of spectral data $\Scal$. are the same as the ones appearing in the coordinatization of the space $\Acal$. Equivalently, one needs to check that the reference sets of eigenvectors defined on the $\Acal$-side by~\eqref{GPsi} and on the $\Scal$-side by~\eqref{G0def}, are the same. This follows from construction of \cite{RH1,Annalen}
of the solution of Riemann--Hilbert problem with quasi-permutation monodromies.
\end{proof}

\section{Variations of Szeg\H{o} kernel and summation over sheets}
\label{secvarSz}

Here we state several properties of the Szeg\H{o} kernel needed in the analysis
of symplectic properties of spectral transform.

The kernel $S_q$ satisfies the following identity due to Fay \cite{Fay73}:
\begin{gather}
\label{Fayid}
S_q(x,y) S_q(y,x)= -B(x,y)- \sum_{\a,\b=1}^g \p_\a\p_\b \log \Theta(q) v_\a(x) v_\b(y),
\end{gather}
where $B(x,y)=\d_x \d_y\log E(x,y)$ is the canonical bimeromorphic differential.

The following formula describes the variation of $S_q$ with respect to $q_\alpha$,
$\a=1,\dots,g$, \cite[Proposition~1]{KalKor}:
\begin{gather*}
\f{\delta}{\delta q_\gamma}S_q(x,y)=-\int_{t\in a_\gamma} S_q(x,t) S_q(t,y)
\end{gather*}

Another fact we need is the following
variational formula for Szeg\H{o} kernel with respect to $I_\a=\int_{a_\a}v$ and residues of $v$ at $t_j^{(k)}$
(formulas of this type appeared in the theory of tau-functions of Witham hierarchies \cite{Krich92} as well
as in the theory of deformations of the spectral curve of the Hitchin's systems \cite{Balduzzi,DonMar}).

\begin{Proposition}\label{varSz}
Variation of the Szeg\H{o} kernel with respect to period coordinates $I_\a=\int_{a_\a}$ and coordinates
$\mu_j^{(k)}$ $($related to residues of $v$ via $\eqref{lmu1})$ on the moduli space of spectral covers is given by the following sum of residues at the branch points:
\begin{gather}\label{SI}
\f{\delta}{\delta I_\a}S_q(x,y)\Big|_{z(x),z(y)}=-\f{\pi {\rm i} }{2}\sum_{j=1}^{p} {\rm res}|_{t=x_j}\f{v_\a(t)\,W_t\bigl[S_q(x,t),\, S_q(t,y)\bigr]}{\d z(t)\,\d y(t)}
\end{gather}
and
\begin{gather}\label{SI0}
\f{\delta}{\delta \bigl( 2\pi {\rm i} \mu_j^{(k)}\bigr)}S_q(x,y)\Big|_{z(x),z(y)}=
-\f{\pi {\rm i} }{2}\sum_{j=1}^{p} {\rm res}|_{t=x_j}\f{w_{t_j^{(k)},t_j^{(k-1)}}(t)\,W_t\bigl[S_q(x,t),\, S_q(t,y)\bigr]}{\d z(t)\,\d y(t)},
\end{gather}
where the $z$-coordinates of the points $x$ and $y$ are kept constant under differentiation and $W_t$ is the Wronskian determinant with respect to variable $t$, $p=n(n-1)(m-2)$ $\eqref{bp}$ is the number of branch points.
Here $w_{x,y}(t)$ denotes the differential of the third kind with residues $\pm 1$ at $x$ and $y$, $w_{x,y}(t)$ is normalized by condition
$\int_{t\in a_\a} w_{x,y}(t)=0$ where it is assumed that the choice of the cycle $a_\a$ in $H_1(\CC\setminus \{x,y\})$ is the same as the one used in the computation of the period $I_\a$.
\end{Proposition}
\begin{proof}
The proof is based on the variational formula for the Szeg\H{o} kernel on the moduli space of meromorphic Abelian differentials (see \cite[Section~2.2 and Proposition~2]{KalKor}, based on formalism of~\cite{JDG,Annalen}).
Namely, denote by $\Hcal_g^k$ the moduli space of pairs $(\Ccal,v)$, where $\Ccal$ is a Riemann surface of genus $g$ and $v$ is an abelian differential with $k$ simple poles at $y_1,\dots,y_k$ and $r=2g-2+k$ simple zeros at $x_1,\dots,x_r$. The homological coordinates on $\Hcal_g^k$ are defined as periods
of $v$ over a~basis $\{s_i\}$ in the relative homology group $H_1\bigl(\CC\setminus \{y_i\}_{i=1}^k, \, \{x_j\}_{j=1}^r\bigr)$. We denote the dual basis in the dual relative homology group $H_1\bigl(\CC\setminus \{x_j\}_{j=1}^r,\,
\{y_i\}_{i=1}^k\bigr)$ by $\{s_i^*\}$ with the intersection index $s_i^*\circ s_j=\delta_{ij}$.
Then variational formulas from \cite{KalKor} look as follows:
\begin{gather}\label{varS}
\f{\delta}{\delta \bigl(\int_{s_i} v \bigr)}{S_{q}(x,y)}= \f{1}{4}\int_{t\in s_i^*} \f{W_t\bigl[S_{q}(x,t),S_{q}(t,y)\bigr]}{v(t)},
\end{gather}
where $W_t[f(t),g(t)]=f'g-fg'$ is the Wronskian, and the differentiation is performed keeping $\int_{x_1}^x v$ and
 $\int_{x_1}^y v$ constant for some choice of the ``first'' branch point $x_1$.

On the other hand, according to the approach of \cite{IMRN}, the space of spectral curves
with simple branch points (the assumption of simplicity can be omitted if necessary, as in the ${\rm SL}(2)$ case) can be naturally imbedded in the space $\Hcal_g^k$ with $k=nm$. On the space of spectral curves we
identify $v$ with the differential $y\,\d z$. The poles $y_j$ of $v$ are identified with points $t_j^{(k)}$ for all
$j=1,\dots,m$ and $k=1,\dots,n$.
Moreover, the $b$-periods of $y\,\d z$, as well as integrals of $y\,\d z$ between branch points become
dependent on its $a$-periods $\{I_\a\}_{\a=1}^g$ and residue variables $\mu_j^{(k)}$. Then the formulas \eqref{SI} and \eqref{SI0} can be derived from
\eqref{varS} in complete analog to the proof of Theorem~4.3 of \cite{IMRN} dealing with variation of holomorphic abelian differentials.
\end{proof}

Finally, we shall need the following.

\begin{Lemma}\label{lemma51}
Let the points $x$, $y$ and $z^{(j)}$ for some $x,y\in \CC$ and $z\in \C$ lie in the same fundamental polygon of $\CC$.
Then
the following identity holds:
\begin{gather*}
\sum_{i=1}^n S_q\bigl(x,z^{(i)}\bigr)S_q\bigl(z^{(i)},y\bigr)= S_q(x,y)\left(\f{1}{z(x)-z}-\f{1}{z(y)-z}\right)\,\d z.
\end{gather*}
\end{Lemma}
\begin{proof}
The left-hand side is the 1-form in $z$ depending only on the point of the base. It has simple poles at
$z=z(x)$ and $z=z(y)$. The coefficient can depend only on $x$ and $y$
and must be~$S_q(x,y)$ due to the singularity structure at $x=z^{(i)}$ and $y=z^{(i)}$.
\end{proof}

\section{Spectral transform as symplectomorphism}

Here we prove that the spectral transform is a symplectomorphism, and, moreover, explicitly compute the
symplectic potential $\theta_\Acal$ in terms of canonical coordinates on the space $\Scal$.

\begin{Theorem}
The symplectic potential $\theta_\Acal$ $\eqref{thA}$, where $G_j$ are given by
$\eqref{Gj}$ and $\eqref{G0def}$, takes the following form in terms of spectral variables:
\begin{gather*}
\theta_{\Acal}= \sum_{\a=1}^g I_\a \delta q_\a + \sum_{j=1}^m\sum_{k=1}^{n-1} \mu_j^{(k)}\delta \rho_j^{(k)}.
\end{gather*}
\end{Theorem}
\begin{proof}
Let us first observe that the symplectic potential $\theta_\Acal$ (\ref{thA}), where $G_j$ are given by
\eqref{Gj} and \eqref{G0def},
can be expressed via the matrices $\Psih(z)$ \eqref{psih} and $Y(z)$ \eqref{Lz} by
residue theorem as follows:
\begin{gather*}\label{thA1}
\theta_{\Acal} = \sum_{j=1}^m \res_{z=t_j} \tr\bigl(Y \Psih^{-1} \delta \Psih\bigr)\,\d z +\sum_{j=1}^m\sum_{k=1}^{n-1} \mu_j^{(k)}\delta \rho_j^{(k)}.
\end{gather*}

The definition \eqref{psih} of $\Psih$ can be written as
\begin{gather*}
\Psih_{ab}(z,z_0)= \psi\bigl(z^{(b)}, z_0^{(a)}\bigr),
\end{gather*}
where $\psi(x,x_0)$ is expressed via Szeg\H{o} kernel via \eqref{psidef}.

Due to \eqref{Psihi}, we have
\begin{gather*}
 (\Psih^{-1})_{ab}= \psi\bigl(z_0^{(b)}, z^{(a)}\bigr).
\end{gather*}

Consider now three different types of contributions to $\theta_{\Acal}$.

{\bf Contribution of $\boldsymbol{\delta {q_\a}}$.} We have
\begin{gather*}
\bigl(\Psih^{-1}(z) \delta_{q_\a} \Psih (z)\bigr)_{aa} =\f{(z-z_0)^2}{\d z\,\d z_0} \sum_{b=1}^n \oint_{t\in a_\gamma} S_q\bigl( z^{(a)} ,t\bigr)S_q\bigl(t,z_0^{(b)}\bigr) S_q\bigl(z_0^{(b)}, z^{(a)}\bigr)
\overset{\text{Lemma \ref{lemma51}}}{=}\\
\phantom{\bigl(\Psih^{-1}(z) \delta_{q_\a} \Psih (z)\bigr)_{aa}}= \f{(z-z_0)^2}{\d z}\oint_{t\in a_\a} S_q\bigl( z^{(a)} ,t\bigr)S_q\bigl(t, z^{(a)} \bigr) \left(\f{1}{z(t)-z_0}-\f{1}{z-z_0}\right)
\end{gather*}

Recall now that $y\colon \CC\to \C$ is a meromorphic function with simple poles at the $mn$ points above $t_j^{(a)}$. The singular parts are
\[
y (x) = \frac{\l_j^{(k)} }{z(x)-t_j} \qquad \mbox{as } x\to t_j^{(a)}.
\]

The contribution of $\delta q_\a$ to $\theta_\Acal$ is therefore given by
\begin{gather*}
\theta_{\Acal}(\delta_{q_\a}) \!=\!
\sum_{j=1}^m \sum_{a=1}^n {\rm res}|_{x=t_j^{(a)}} \!\left[y(x)(z(x)-z_0)^2
\oint_{t\in a_\a}\!\!S_q( x ,t)S_q(t, x) \bigg(\f{1}{z(t)-z_0}-\f{1}{z(x)-z_0}\bigg)\right]\\
\phantom{\theta_{\Acal}(\delta_{q_\a}) }
\!=\!\sum_{j=1}^m \sum_{a=1}^n \res_{x=t_j^{(a)}} y(x)\oint_{t\in a_\a}S_q( x ,t)S_q(t, x)
 \left(\f{(z(x)-z_0)^2}{z(t)-z_0}-(z(x)-z_0)\right).
\end{gather*}
The integration contours $a_\a$ in the $t$-variable can be chosen so as not to intersect the integration contours for the residues in the $x$-variable and hence the integrand is regular. Using Fubini's theorem, we can exchange the order of integrations
\begin{gather*}
\theta_{\Acal}(\partial_{q_\a}) =\oint_{t\in a_\a} \sum_{j=1}^m \sum_{a=1}^n \res_{x=t_j^{(a)}} y(x)S_q( x ,t)S_q(t, x ) \bigl(z(x)-z(t)\bigr)\,\f{z(x)-z_0}{z(t)-z_0}.
\end{gather*}

The sum over the residues is the sum over all poles above the $t_j$'s of the differential (in the~$x$ variable)
\begin{gather}\label{diff}
y(x)S_q( x ,t)S_q(t, x ) \bigl(z(x)-z(t)\bigr)\,\f{z(x)-z_0}{z(t)-z_0}.
\end{gather}
Due to (\ref{Fayid}), this differential has only one additional simple pole at $x=t$ (taking into account the asymptotics of $B(x,t)$ on the diagonal), with residue
\begin{gather*}
\res_{x=t}\left[ y(x)S_q( x ,t)S_q(t, x ) \bigl(z(x)-z(t)\bigr)\,\f{z(x)-z_0}{z(t)-z_0} \right]= - y(t)\,\d z(t) .
\end{gather*}
The differential \eqref{diff} does not have poles at $x=\infty^{(a)}, \ a=i,\dots,n$ because, even if $z(x)$ has simple poles at these points, the function $y(x)$ has double zeros there (such that $v=y\,\d z$ is regular at infinities).

Thus we can use Cauchy's residue theorem and get
\begin{gather*}
\theta_{\Acal}(\pa_{q_\a}) = \oint_{ a_\a} y\,\d z = I_\a.
\end{gather*}

{\bf Contribution of $\boldsymbol{\delta_{I_\a}}$.}
Using variational formulas (\ref{SI}), we get
\begin{gather*}
\bigl(\Psi^{-1}\delta_{I_\a}\Psi\bigr)_{aa}(z)
=\sum_{c=1}^n \frac{S_q\bigl(z_0^{(c)} ,x\bigr)}{\d z(x)}
\left(\f{\pi {\rm i}}{2}\sum_{x_i}
{\rm res} |_{t=x_i} \f{v_\a(t)W_t\bigl[S_q(x,t), S_q\bigl(t,z_0^{(c)}\bigr)\bigr]}{\d y(t)\,\d z(t)}\right)\bigg|_{x=z^{(a)}}.
\end{gather*}

Therefore,
\begin{gather*}
\theta_{\Acal}(\pa_{I_\a}) =
\sum_{j=1}^m \sum_{a=1}^n \l_j^{(a)} \bigl(\Psi^{-1}\delta_{I_\a}\Psi\bigr)_{aa}\\
=
\f{\pi {\rm i}}{2}\sum_{\ell=1}^m \sum_{a=1}^n \res_{x=t_\ell^{(a)}} \le(\sum_{x_i\in b.pts} \res_{t=x_i} y(x)\frac {v_\a(t) W_t\bigg[S_q(x,t), S_q(t,x) \f{\bigl(z(x)-z(t)\bigr)\bigl(z(x)-z_0\bigr)}{z(t)-z_0 }\bigg]}{\d y(t)\,\d z(t) }\ri).
\end{gather*}
Once again, we can swap the order of residues because the branch-points are away from the points $t_j^{(a)}$.
One can verify that the differential of $x$ in the inmost residue has poles only above the $t_j$'s and at $x=t$. The residue at $x=t$ can be computed using the following relation which can be verified using the local expansion of $S_q(x,t)$ near the diagonal (one needs to use the coordinate $z(x)$ and observe that the second argument of the Wronskian is non-singular at~$x=t$):
\begin{gather*}
\res_{x=t} \left(y(x)W_t\bigg[S_q(x,t), S_q(t,x) \f{\bigl(z(x)-z(t)\bigr)\bigl(z(x)-z_0\bigr)}{z(t)-z_0 } \bigg]\right)
= \d y(t)\,\d z(t).
\end{gather*}

Therefore,
\begin{gather*}
\theta_{\Acal}(\delta_{I_\a})
=\f{\pi {\rm i}}{2} \sum_{x_i} \res_{t=x_i} v_\a(t) =0.
\end{gather*}

{\bf Contribution of $\boldsymbol{\delta \mu_j^{(k)}}$.} To compute this contribution to $\theta_\Acal$, one can repeat the previous computation using (\ref{SI0}). This boils down to replacing $v_\a$ by $w_{t_j^{(k)},t_j^{(k-1)}}$.
The end result is
\begin{gather*}
 \theta_{\Acal}(\delta_{ \mu_j^{(k)}})= \f{\pi {\rm i}}{2} \sum_{x_i} \res_{t=x_i} w_{t_j^{(k)},t_j^{(k-1)}},
\end{gather*}
which vanishes due to the genericity assumption that none of the branch points $x_j$ project to the poles $t_k$.
\end{proof}

The following corollary provides a set of canonical Darboux coordinates on the space $\Acal$.

\begin{Corollary}\label{corsymp}
The symplectic form $\omega_\Acal$ on the space $\Acal$ coincides with the $($pullback under the spectral transform$)$ of
the symplectic form $\omega_\Scal$ on the space of spectral data $\Scal$ $\eqref{SP1}$.
Therefore, $\omega_\Acal$
can be written in terms of
spectral variables $\eqref{fullco}$ as follows:
\begin{gather}\label{oAS}
\omega_{\Acal}= \sum_{\a=1}^g \delta I_\a \wedge \delta q_\a
+ \sum_{j=1}^m\sum_{k=1}^{n-1} \delta\mu_j^{(k)} \wedge \delta\rho_j^{(k)}.
\end{gather}
\end{Corollary}

Notice that this corollary gives the set of Darboux coordinates for the symplectic form $\omega_\Acal$, which
is defined on a bigger space than the symplectic forms considered in \cite{AHH} or \cite{BBT}.
Moreover, even if one ignores this extension and restrict (\ref{oAS}) to symplectic leaves of the usual Kirillov--Kostant bracket, we get the form
\begin{gather*}
\omega^0_{\Acal}= \sum_{\a=1}^g \delta I_\a \wedge \delta q_\a,
\end{gather*}
i.e., $(I_\a,q_\a)$ are Darboux coordinates for the symplectic form $\omega_{\Acal_0}^L$ on the
symplectic leaf of the space $\Acal$. It is this form which was studied in previous works \cite{BBT}.

However $\bigl\{I_\a,q_\a\bigr\}_{\a=1}^g$ are not the previously known Darboux coordinates on $\omega_{\Acal_0}^L$ described in say \cite[Section 5.11]{BBT}:
while the periods $I_\a$ are the same, the coordinates $q_\a$ differ from the ones of \cite{BBT}
by components of the vector of Riemann constants. To discuss the link of our scheme to the well-known textbook results we need to analyze the dependence of this vector of Riemann constants on moduli.

\section{Variations of vector of Riemann constants}\label{VarKsec}

Here we describe variations of the vector of Riemann constants with respect to moduli $I_\a$ of the spectral curve
by combining results of \cite{IMRN,JDG}.

To formulate these formulas, we need to define the following object, denoted $\s(x,y)$ \cite{Fay92} (see also~\cite[formula~(2.15)]{JDG}):
\begin{gather*}
\s(x,y)= \left(\f{c(x)}{c(y)}\right)^{1/(1-g)}.
\end{gather*}
Here $c(x)$ is the following fundamental (multi-valued) $g(1-g)/2$-differential on $\CC$:
\begin{equation*}
c(x)=\frac{1}{{\mathcal W}[v_1, \dots, v_g](x)}\sum_{\a_1, \dots,
\a_g=1}^g
\frac{\partial^g\Theta(K^x)}{\p z_{\a_1}\cdots \p z_{\a_g}}
v_{\a_1}\cdots v_{\a_g}(x),
\end{equation*}
where
\begin{gather*}
{\mathcal W}[v_1, \dots, v_g](x)= {\rm \det}_{1\leq \a, \b\leq g}\bigl\|v_\b^{(\a-1)}(x)\bigr\|
\end{gather*}
is the Wronskian determinant of holomorphic differentials
 at the point $x$.

\begin{Proposition} Let $x\in \CC$ be a point such that $\pi(x)$ is independent of coordinates $I_\a$ and $\mu_j^{(k)}$ on the moduli space of spectral curves.
 Then the following variational formulas hold:
\begin{gather}\label{varKI}
\f{\delta K^x_\b}{\delta I_\a}\Big|_{z(x)}= \sum_{j=1}^{p} {\rm res}|_{t=x_j}\f{v_\a(t)\,v_\b(t) \d_t\log\frac{\s(t,x_0)E(t,x)^{g-1}}{\sqrt{v(t)}} }{\d z(t)\,\d y(t)}
\end{gather}
and
\begin{gather}\label{varKmu}
\f{\delta K^x_\a}{\delta \mu_j^{(k)}}\Big|_{z(x)}=\sum_{j=1}^{p} {\rm res}|_{t=x_j}\f{v_\a(t)\, w_{t_j^{(k)},t_j^{(k-1)}}(t)\,\d_t\log\frac{\s(t,x_0)E(t,x)^{g-1}}{\sqrt{v(t)}} }{\d z(t)\,\d y(t)}.
\end{gather}
\end{Proposition}
\begin{proof}
The formula (2.58) of \cite{JDG} gives the variation of $K^x$ on the space of
holomorphic abelian differentials; the straightforward extension to spaces of
meromorphic abelian differentials with simple poles described in \cite{KalKor} gives the formula, which in notations of the proof of Proposition~\ref{varSz} look as follows:
\begin{gather}
\label{varkp1}
\f{\p K^x_\b}{\p \bigl(\int_{s_j}v \bigr)}\Big|_{z(x)}=\f{1}{2\pi {\rm i}}\oint_{s_j^*}
\f{v_\b(t)}{v(t)}\d_t\log\frac{\s(t,x_0)E(t,x)^{g-1}}{\sqrt{v(t)}}.
\end{gather}
Similarly to the proof or Proposition \ref{varSz}, the variational formulas \eqref{varKI} and \eqref{varKmu}
follow from~\eqref{varkp1} by the application of the formalism of \cite{IMRN}.
\end{proof}

The formula \eqref{varKI} implies the following immediate corollary which we are going to need below.

\begin{Corollary}\label{symmK}
Let $x\in \CC$ be a point such that $\pi(x)$ is independent of coordinates $I_\a$. Then the vector of Riemann constants satisfies
\be
\f{\delta K^x_\a}{\delta I_\b}=\f{\delta K^x_\b}{\delta I_\a}.
\label{grad}
\ee
\end{Corollary}

Therefore, there exists a potential $F^x$ satisfying
\be
\f{\delta F^x}{\delta I_\a}= K^x.
\label{potK}
\ee
At the moment we don't know how to compute $F^x$ explicitly outside of hyperelliptic locus.

\section{Reduction to symplectic leaves of Kirillov--Kostant bracket}

Let us now discuss the relation of our results to the existing results concerning the Darboux coordinates
on the symplectic leaves of the Kirillov--Kostant Poisson structure (we refer to \cite{BBT} and references therein
for systematic description of previous results).

Consider the space $\Acal_0$ defined as
\begin{gather*}
\Acal_0=\Biggl\{\{A_j\}_{j=1}^m,\, \sum_{j=1}^m A_j=0 \Biggr\}\Big/{\sim},
\end{gather*}
where $A_j\in \mathfrak{sl}(n)$ and the equivalence relation $\sim$
identifies the sets $A_j$'s which differ by a~simultaneous conjugation $A_j\to S A_j S^{-1}$ with $S\in {\rm SL}(n)$.

The Poisson structure on $\Acal_0$ is given by the Kirillov--Kostant bracket for each $A_j$,
\begin{gather}\label{KK}
\bigl\{A_j^a, A_j^b\bigr\}= f^{ab}_c A_j^c
\end{gather}
with all other brackets vanishing.
The bracket \eqref{KK} is induced by the bracket \eqref{KKex} on the space~$\Acal$ under the natural map $\Acal\to \Acal_0$ defined by $A_j=G_j L_j G_j^{-1}$.

In contrast to the bracket \eqref{KKex} which is non-degenerate on $\Acal$, the bracket \eqref{KK} on $\Acal_0$ is degenerate with Casimirs given by the eigenvalues $\l_j^{k}$ of matrices $A_j$.
On symplectic leaves $\Acal_0^L$ defined by $L_j={\rm const}$, the Poisson structure can be inverted. This gives the symplectic form~$\omega_{\Acal_0}^L$ defined by the first term in \eqref{sf1},
\begin{gather}\label{sf0}
\omega_{\Acal_0}^L =
\tr \sum_{j=1}^m \delta G_j G_j^{-1}\wedge L_j \delta G_j G_j^{-1} .
\end{gather}
The form \eqref{sf0} is independent of the choice of eigenvectors $G_j$ of matrices $A_j$.

The space of spectral data corresponding to the space $\Acal_0^L$ is obtained from the space $\Scal$ \eqref{SP1} by omitting the variables $R_j$,
\begin{gather*}
\Scal=\bigl\{ \{Q_j\}_{j=2}^n,\, q\in \C^g\bigr\}.
\end{gather*}
On symplectic leaves $\Scal^L$ in $\Scal$, defined by conditions ${\rm res}|_{t_k^{(j)}}={\rm const}$ for all $k$ and $j$, the canonical symplectic form is obtained by restriction of $\omega_S$ \eqref{oS},
\begin{gather*}
\omega_{\Scal}^L=\sum_{\a=1}^g \delta I_\a\wedge \delta q_\a  .
\end{gather*}

Now Corollary \ref{corsymp} implies that the pullback of the form $\omega_{\Scal}$ under the spectral transform
$\Acal_0^L\to\Scal^L$ coincides with $\omega_{\Acal_0}^L $.

\begin{Proposition} \label{probdar}
The periods $I_\a$ and the components $q_\a$ of the vector $q$ are Darboux on the
symplectic leaf of the space $\Acal_0$,
\begin{gather}\label{omdd}
\omega_{\Acal_0}^L=\sum_{\a=1}^g \delta I_\a\wedge \delta q_\a.
\end{gather}
\end{Proposition}

The Darboux coordinates given by proposition (\ref{probdar}) are similar to the known action-angle variables ($(I_i,\theta_i)$ variables in (5.75) of \cite{BBT}). However, they are not really the same. Namely, the actions $I_\a$ are the same while the coordinates $q_\a$ differ from $\theta_i$ of \cite{BBT}. To deduce \cite[formula~(5.75)]{BBT} from \eqref{omdd}, one needs to use
Corollary \ref{symmK} which implies the proposition.
\begin{Proposition}
Define the vector $\ang^x\in \C^g$ by
\[
\ang^x=q+ K^{x}
\]
for any $x$ such that $\pi(x)$ remains constant under variation of moduli $($in particular, one can choose $x=z_0^{(j)}$ for any $j$$)$.
Then the form $\omega_{\Acal_0}^L$ can also be written as
\[
\omega_{\Acal_0}^L=\sum_{\a=1}^g \delta I_\a\wedge \delta \ang^x_\a.
\]
\end{Proposition}
\begin{proof}Since the vector of Riemann constants depends only on the moduli $I_\a$ of the spectral curve, we get
\[
\sum_{\a=1}^g \delta I_\a\wedge \delta \ang^x_\a=\sum_{\a=1}^g \delta I_\a\wedge \delta q_\a+
\sum_{\a=1}^g \delta I_\a\wedge \delta K^{x}_\a .
\]
The last sum equals
\[
\sum_{\a<\b} \left(\f{\delta K^{x}_\a}{\delta I_\b}- \f{\delta K^{x}_\b}{\delta I_\a} \right) \delta I_\a\wedge \delta I_\b
\]
which vanishes due to \eqref{grad}.
\end{proof}

The variables $(I_\a,\ang_\a^x)$ are the traditional action-angle variables described in \cite[Secton~5.11]{BBT} ($\ang_\a^x$~are denoted there by $\theta_i$).
The advantage of the variables $q_\a$ over $\ang_\a^x$ is that the former are intrinsic, i.e., the vector $q$ is determined uniquely by a point of the space $\Acal$
(in particular, $q$~independent of he choice of the point $z_0$ used in our explicit construction of the direct spectral transform). On the other hand, the vector $\ang^x=q+K^x$ depends on the point $x$ (as long as the projection $z(x)$ remains moduli-independent), i.e., there is the one-parametric family of the angle variables $\gamma_\a^x$ which are related by
\[
\ang^{{\tilde x}}=\ang^x+ (g-1)\Ucal_x({\tilde x}).\]
On the symplectic leaves $\Acal_0^L$ within the space $\Acal_0$ both sets $(I_\a,\,q_\a)$ are $(I_\a,\,\ang^x_\a)$
give Darboux coordinates for $\omega^L_{\Acal_0}$. However, on the extended space $\Acal$ this is not the case any more. Since the vector $K^x$ non-trivially depends on eigenvalue coordinates \smash{$\mu_j^{(k)}$}, in the
set of Darboux coordinates on~$\Acal$ given by $\bigl(I_\a,\,q_a,\,\mu_j^{(k)}\,\rho_j^{(k)}\bigr)$ one can not replace $q_\a$'s by
$\ang^x_\a$. The symplectic form $\omega_\Acal$ in the coordinates $\bigl(I_\a,\,\ang^x_\a,\,\mu_j^{(k)}\,\rho_j^{(k)}\bigr)$ does not have canonical form, or even constant coefficients.

On the symplectic leaf $\Acal_0^L$, where both sets $(q_\a,I_\a)$ and $(\ang^x_\a,I_\a)$ provide Darboux coordinates, the generating function between them can be defined by
the difference of the corresponding symplectic potentials
\[
\delta F^x= \sum_{\a=1}^g \ang^x_\a \delta I_\a - \sum_{\a=1}^g q_\a \delta I_\a= \sum_{\a=1}^g K_\a \delta I_\a.
\]
The function $F^x$ is therefore nothing but the potential \eqref{potK} for the vector of Riemann constants
which we do not know in general. However, $F^x$ can be easily computed in the hyperelliptic case.

\begin{Remark}
Symplectic properties of spectral transform for various operators (both finite-dimensional and infinite-dimensional) were extensively studied in the theory of integrable systems and summarized in various textbooks. We refer here to \cite{BBT,FadTak}.

We would like to point out some connection (or lack thereof) of our results to previous ones.
First, we remind that the form
\begin{gather}\label{omega10}
\omega=\operatorname{tr} \delta G G^{-1} \wedge L \delta G G^{-1}
\end{gather}
is the symplectic form on the symplectic leaf $L={\rm const}$ of the Poisson--Lie--Kirillov--Kostant--(Souriau) bracket on the space of matrices $A=GLG^{-1}$ (with $L$ the diagonal form of $A$). The 1-form
\[
\theta= \operatorname{tr} L \, G^{-1} \delta G
\] is the corresponding symplectic potential. This relation between the form \eqref{omega10} and the PLKKS bracket has been known since the 1993 paper \cite{AlekMalk}.

 The form \eqref{omega10} was later used as a building block of various symplectic forms in
 the theory of finite and infinite-dimensional integrable systems. For example, in the paper by Krichever and Phong \cite{KricheverPhong} both $G$ and $L$ were considered as power series or polynomials of the spectral parameter, and the actual symplectic form was given by the residue of the expression \eqref{omega10} at
 infinity. Similarly, in more recent paper by Krichever \cite{Krichever} both $G$ and $L$ were
 depending on a~point of a Riemann surface such that the expression \eqref{omega10} becomes a meromorphic differential there. To get the actual symplectic form this differential needs to be integrated over
 a closed contour (around a pole with residue, or one of homology cycles).

 In this paper the phase space is a direct sum of several copies of elementary phase spaces endowed
 with the form
 \eqref{omega10}. This form becomes symplectic on the full phase space after an~addition of the terms which give the non-degenerate extension of the PLKKS bracket (see the second term in \eqref{omega0}).

 Thus the objects we are dealing with in this paper are different from the ones considered before. Our results on symplectic properties of the spectral transform for rational matrices look similar, but are actually also different from the previous ones (we get a new set of spectral Darboux coordinates). Finally, the methods that we apply here (the variational formulas on Riemann surfaces)
 were not used in this context before. The variational formulas for Szeg\H{o} kernel in fact seem to be new.
\end{Remark}

\subsection*{Acknowledgements}
The work of M.B.\ was supported in part by the Natural Sciences and Engineering Research Council of Canada (NSERC) grant RGPIN-2016-06660.
The work of D.K.\ was supported in part by the NSERC grant
RGPIN-2020-06816.

\pdfbookmark[1]{References}{ref}

\LastPageEnding

\end{document}